\documentclass[11pt,twoside]{article}
\usepackage{amstext}
\usepackage{amsmath} 
\usepackage{amsthm}
\usepackage[dvips]{graphicx}
\usepackage{amsfonts}
\usepackage{amssymb}
\usepackage{amsxtra}
\usepackage{natbib}
\usepackage{multicol}
\usepackage{graphicx}
\usepackage{color}
\usepackage{multirow}
\numberwithin{equation}{section}

\newtheorem{theorem}{Theorem}[section]
\newtheorem{lemma}{Lemma}[section]
\newtheorem{corollary}{Corollary}[section]
\newtheorem{remark}{Remark}[section]
\newtheorem{proposition}{Proposition}[section]

\newcommand{\F}{\mathcal{F}}

\renewcommand{\P}{\mathbb{P}}
\newcommand{\E}{\mathbb{E}}

\newcommand{\Q}{\mathbb{Q}}

\newcommand{\p}{\Phi^{+}}

\newcommand{\R}{\mathbb{R}}

\newcommand{\A}{\mathcal{A}}
\newcommand{\tv}{\tilde{v}}
\newcommand{\zM}{z_{\text{max}}}
\newcommand{\zm}{z_{\text{min}}}
\newcommand{\tP}{\widetilde{P}}

\title{\LARGE \bf
\sc Pricing Asian Options for Jump Diffusions
\thanks{This research is partially supported by NSF Research
Grant, DMS-0604491.} \thanks{Department of Mathematics, University
of Michigan, 530 Church Street, Ann Arbor, MI 48109.}}
\author{Erhan Bayraktar \thanks{ e-mail: erhan@umich.edu}
\and Hao Xing \thanks{ e-mail: haoxing@umich.edu} }

\date{}
\begin{document}
\maketitle
\begin{abstract}
We construct a sequence of functions that uniformly converge (on compact sets) to the price of Asian option, which is written on a stock whose dynamics follows a jump diffusion, exponentially fast.
We show that each of the element in this sequence is the unique classical solutions of a parabolic partial differential equation (not an integro-differential equation).
As a result we obtain a fast numerical approximation scheme whose accuracy versus speed characteristics can be controlled. We analyze the performance of our numerical algorithm on several examples. 
\end{abstract}

\section{Introduction}

We develop an efficient numerical algorithm to price Asian options, which are derivatives whose pay-off depends on the average of the stock price, for jump diffusions. The jump diffusion models are heavily used in the option pricing context since they can capture the excess kurtosis of the stock price returns and along with the skew in the implied volatility surface (see \cite{cont-tankov}). Two well-known examples of these models are i) the model of \cite{merton}, in which the jump sizes are log-normally distributed, and ii) the model of \cite{kou}, in which the logarithm of jump sizes have the so called double exponential distribution. Here we consider a large class of jump diffusion models including these two.

The pricing of Asian options is complicated because it involves solving a partial differential equation (PDE) with two space dimensions, one variable accounting for the average stock price, the other for the stock price itself. However, \cite{vecer} and \cite{vecer-xu} were able to reduce the dimension of the problem by using a change of measure argument (also see Section 2.1). When the stock price is a geometric Brownian motion \cite{vecer} showed that the price of the Asian option at time $t=0$, which we will denote by $S_0 \rightarrow V(S_0)$, satisfies $V(S_0)=S_0 \cdot v(z=z^*,t=0)$ for a suitable constant $z^*$, in which the function $v$ solves a one dimensional parabolic PDE. When the stock price is a jump diffusion, then under the assumptions that $v_t$, $v_z$ and $v_{zz}$ are continuous \cite{vecer-xu} (Theorem 3.3 and Corollary 3.4) showed that the function $v$ solves an integro-partial differential equation using It\^{o}'s lemma. 
However, a priori it is not clear that these assumptions are satisfied. In this paper, we show that for the jump diffusion models these assumptions are indeed satisfied (see Theorem~\ref{theorem-2}), i.e. we directly show that $v$ is the unique classical solution of the integro-partial differential equation in \cite{vecer-xu}.
 We do this by showing that $v$ is the limit of a sequence of functions 
constructed by iterating a suitable functional operator, which we will denote by $J$, that takes functions with certain regularity properties into the unique classical solutions of parabolic differential equations and gives them more regularity and finally showing that $v$ is the fixed point of the functional operator and that it satisfies the certain regularity properties. This proof technique is similar to that of \cite{bayraktar-finite-horizon}, in which the regularity of the American put option prices are analyzed. In the current paper, some major technical difficulties arise because the pay-off functions we consider are not bounded and also because the sequence of functions constructed is not monotonous (\cite{bayraktar-finite-horizon} was able to construct a monotonous sequence because of the early exercise feature of the American options).

The iterative construction of the sequence of functions which converge to the Asian option price naturally leads to an efficient numerical method for computing the price of Asian options. We prove that the constructed sequence of functions converges to the function $v$ uniformly (on compact sets) and exponentially fast.
Therefore, after a few iterations one can obtain the function $v$ to the desired level of accuracy, i.e.
the accuracy versus speed characteristics of the numerical method we propose can be controlled. On the other hand, since each element of the approximating sequence solves a parabolic PDE (not an integro-differential equation), we can use one of the classical finite difference schemes to determine it. We propose a numerical scheme in Section~\ref{sec:num} and analyze the performance it in the same section. 

So far numerical methods for the pricing Asian options were proposed for diffusion models: \cite{vecer}, \cite{rogers}, \cite{zhang} developed PDE methods, \cite{geman-yor} developed a single Laplace inversion method, \cite{linetsky} developed a spectral expansion approach, \cite{cai-kou} developed a double Laplace inversion method. On the other hand \cite{rogers} and \cite{thompson} obtained tight bounds for the prices. We should mention that \cite{cai-kou} also considered pricing Asian option for the double exponential jump diffusion model of \cite{kou}. 

The rest of the paper is organized as follows: In Section 2.1, we summarize the findings of \cite{vecer-xu} in the context of jump-diffusion models. In Section 2.2, we present our main theoretical results: Theorem~\ref{theorem-2} and Corollary~\ref{theorem:main}.
We propose a numerical algorithm in Section 3.1 and analyze the convergence properties of our algorithm in Section 3.2 (see Propositions~\ref{prop:disc-iter-conv} and \ref{prop:disc-conv}). Then we perform numerical experiments for two particular jump diffusion models and analyze the performance of our numerical algorithm in Section 3.3. Section 4 is devoted to development of the proof of Theorem~\ref{theorem-2}. In Sections 4.1 and 4.2 we develop the properties of the functional operator  $J$ and the properties of the sequence obtained by iterating $J$, respectively. These results are used to prove Theorem~\ref{theorem-2} in Section 4.3. A brief summary of our proof technique can be found in Section 2.2. right after the definition of the operator $J$ in (\ref{eq:def-j}).

\section{A Sequential Approximation to Price of an Asian Option}
\subsection{Dimension Reduction}
Let $(\Omega, \F, \P)$ be a complete probability space hosting a Wiener process $\{B_t; t\geq 0\}$
 and a Poisson random measure $N$, whose
mean measure is $\lambda \nu(dy) dt$, independent of the Wiener
process.  Let $(\F_t)_{t\geq 0}$ denote the natural filtration of
$B$ and $N$.  In this filtered probability space, let us define a
Markov process $S=\{S_t; t\geq 0\}$ via its dynamics as
\begin{equation}\label{eq:dyn-S}
dS_t=(r-\mu)S_{t-} \, dt + \sigma S_{t-} \, dB_t + S_{t-} \int_{\R_+}
(y-1)N(dt, dy),
\end{equation}
in which $r$ is the risk free rate, $\mu\triangleq \lambda(\xi-1)$
with assumption $\xi \triangleq \int_{\R_+}y \nu(dy)< \infty$. The process $S$ is the price of a traded stock, and under the measure $\P$, the discounted stock price $\left(e^{-rt}S_t\right)_{t \geq 0}$ is a martingale. In this framework the stock price jumps from at time $t$ the stock price moves from $S_{t-}$ to $ S_{t-} Y$, in which $Y$'s distribution is given by $\nu$.  $Y$ is a positive random variable and note that when $Y <1$ then the stock price S jumps down when $Y>1$ the stock price jumps up. In Merton's jump diffusion model (see \cite{merton}) $Y=\exp(X)$ where $X$ is a Gaussian random variable. In Kou's model (\cite{kou}) $Y=\exp(X)$ in which $X$ has the double exponential distribution.

To reduce the dimension of the Asian option pricing problem, \cite{vecer-xu} introduce
 a new measure $\Q$ by
\begin{equation}
\frac{d \mathbb{Q}}{d
\mathbb{P}}\bigg|_{\mathcal{F}_t}=e^{-rt}\frac{S_t}{S_0}, \quad t
\in [0,T].
\end{equation}
Here, $T$ is the maturity of the Asian option.
\cite{vecer-xu} also introduce a numeraire process
\begin{equation}\label{eq:def-Zt}
Z^J_t \triangleq \frac{X_t}{S_t}, \quad t \in [0,T],
\end{equation}
where $X=\{X_t; t\in [0,T]
\}$ is a self-financing portfolio, which replicates the pay-off of the Asian option, whose dynamics are given by
\begin{equation}
dX_t =q_t \, dS_t + r(X_{t-} -q_t \,S_{t-})\,dt, \quad X_0=x,
\end{equation}
in which $q_t$ defined as
\begin{equation}\label{eq:def-q}
 q_t \triangleq \frac{1}{rT}(1-e^{-r(T-t)}), \quad t\in [0,T],
\end{equation}
is the number of shares invested in the stock, and
\begin{equation} \label{eq:def-x0}
x=q_0 S_0 - e^{-rT}K_2.
\end{equation}

\cite{vecer-xu} showed that the price of the continuous
averaging Asian option with floating strike $K_1$ and fixed strike
$K_2$ defined by
\begin{equation}\label{eq:def-asian-price-P}
V(S_0)\triangleq
\E^{\P}\left\{e^{-rT}\left(\zeta\cdot\left(\frac1T \int_0^T S_tdt
- K_1S_T - K_2\right)\right)^+\right\}
\end{equation}
can also be represented as
\begin{equation}\label{eq:def-asian-price-Q}
V(S_0) = S_0 \cdot
\E^{\Q}_{0,Z_0^{J}}[\left(\zeta\cdot(Z^J_T-K_1)\right)^+],
\end{equation}
in which $\zeta \in \{-1,1\}$ indicates whether the option is a
put or a call. Throughout this paper, the short hand notation
$\E^{\Q}_{t,z}$ represents the conditional expectation under $\Q$, given
the process at time $t$ is $z$. Under the measure $\Q$, the Poisson random measure will
have mean measure $\lambda\tilde{\nu}(dy)dt$ with new jump
measure $\tilde{\nu}(dy)=y\nu(dy)$.

\subsection{Main Theoretical Results}

In this section we will show that
\begin{equation}\label{eq:lookfv}
V(S_0)=S_0 \cdot v(Z_0^{J},0),
\end{equation}
for some $v$ that is the classical solution,  i.e. $v \in C^{2,1}$, of an integro-partial differential equation (see Theorem~\ref{theorem-2} and Corollary~\ref{theorem:main}). We will also show that $v$ is the limit of a sequence of functions constructed by iterating a functional operator, which is defined in (\ref{eq:def-j}). We will show that each of the functions in this sequence are classical solutions of partial differential equations (not integro-differential equations) and that they converge to $v$ locally uniformly and exponentially fast (see Theorem~\ref{theorem-2}). The analytical properties of the functional operator (listed in the lemmas of Section~\ref{sec:propJ}) used in the construction of the approximating sequence 
play important roles in proving our main mathematical result. We will summarize the role of the functional operator below after we introduce it.

Let us introduce  the following sequence of functions
\begin{equation}\label{eq:def-vn}
v_0(z,t) \triangleq \left(\zeta\cdot(z-K_1)\right)^+, \quad
v_{n+1}(z,t) \triangleq Jv_n(z,t) \quad n\geq 0, \quad \text{for all }
(z,t)\in \R\times[0,T],
\end{equation}
in which the functional operator $J$ is defined, through its action on a test
function $f: \R \times [0,T] \rightarrow \R_+$, as follows:
\begin{equation}\label{eq:def-j}
\begin{split}
Jf(z,t) &=
\E^{\Q}_{t,z}\left\{e^{-\lambda\xi(T-t)}\left(\zeta\cdot\left(Z_T -K_1\right)\right)^+
+ \int_t^T e^{-\lambda\xi(s-t)}\lambda \cdot Pf(Z_s, s) \,ds
\right\},
\end{split}
\end{equation}
in which
$Z = \{Z_t; t\geq 0\}$ has the
dynamics
\[
 dZ_t = -\mu \left(q_t - Z_t\right) dt + \sigma \left(q_t - Z_t \right)dW_t
\]
 and
\begin{equation}\label{eq:op-P}
\begin{split}
Pf(Z_t,t) =& \int_{\R_+} f\left(Z_t +
(q_t-Z_t)\frac{y-1}{y},t\right)y \nu(dy)\\
=& \int_{\R_+} f\left(\frac{Z_t}{y}+q_t \frac{y-1}{y},t\right)y
\nu(dy).
\end{split}
\end{equation}

We will show that the sequence of functions defined in (\ref{eq:def-vn}) by iterating $J$ are classical solutions of PDEs thanks to the following analytical properties of the operator $J$  (which are developed in Section~\ref{sec:propJ}) :
1) $J$ maps functions that are Lipschitz continuous with respect to the $z$-variable (uniformly in the $t$-variable) and H\"{o}lder continuous with respect to the $t$-variable into classical solutions of PDEs (see Proposition~\ref{theorem-1}), 2)  $J$ preserves Lipschitz continuity with respect to the $z$-variable (see Lemma~\ref{lemma:lip}), 3) $J$ transforms Lipschitz continuous functions, with respect to the $z$-variable, that satisfy a linear (in the $z$-variable) growth condition (uniformly in the $t$-variable) into H\"{o}lder continuous functions of the $t$-variable (see Lemma~\ref{lemma:holder}), 4) $J$ preserves the linear growth condition in the $z-$variable (see Lemma~\ref{lemma:mjf} and Remark~\ref{remark-2.2}).
The analytical properties of $J$ can be summarized as  ``J maps nice functions (set of functions with 
a few regularity properties), to nicer 
functions (set of functions that are the classical solutions of partial differential equations, and have the same regularity properties as before).

It is a priori not clear that  the sequence of functions defined in (\ref{eq:def-vn}) has a limit. Using the properties of the operator $J$ we show that this sequence is Cauchy (see Lemma~\ref{lemma:cauchy}) and therefore has a limit (in fact the sequence converges locally uniformly and exponentially fast). We show that, the limit of this sequence, which we denote by $v_{\infty}$, is a classical solution of an integro-PDE using 1)
 the fact that it is a fixed point of the operator $J$ (see Lemma~\ref{lemma:fixpt}), 2) the facts that
it is Lipschitz continuous with respect to the $z$-variable (uniformly in the $t$-variable) (see Lemma~\ref{lemma:vinf-lip}) and H\"{o}lder continuous with respect to the $t$-variable (see Corollary \ref{lemma:vinf-holder}). 

Finally, using a verification argument we will show that the limit $v_{\infty}$ is indeed the function that satisfies (\ref{eq:lookfv}) (see Corollary~\ref{theorem:main}).

The main theoretical results that are summarized above will be stated in the next theorem and its corollary. 
The proof of Theorem~\ref{theorem-2} is given in Section~\ref {sec:proof}
 which uses the results of Sections~\ref{sec:propJ} and \ref {sec:analyze-vn} as summarized above.

\begin{theorem}\label{theorem-2}
(i) The sequence of functions defined in (\ref{eq:def-vn}) has a pointwise limit. 
Let us denote this limit by $v_{\infty}(z,t)$. 

(ii)For any compact domain $\mathcal{D}\subset \R$, $v_n(z,t)$
converges uniformly to $v_{\infty}(z,t)$ for $(z,t)\in
\mathcal{D}\times[0,T]$. Moreover,
\begin{equation}\label{eq:uniform-conv}
|v_{\infty}(z,t)-v_n(z,t)|\leq
M_{\mathcal{D}}\left(1-e^{-\lambda\eta(T-t)}\right)^n,
\end{equation}
where $M_{\mathcal{D}}$ is a constant depending on $\mathcal{D}$
and $\eta=\max\{\xi,1\}$.

(iii) For $n \geq 0$,
the function $v_{n+1}$ is the unique classical solution, i.e. $v_{n+1} \in C^{2,1}$, of
\begin{eqnarray}
& &\A(t)v_{n+1}(z,t) -\lambda\xi v_{n+1}(z,t) + \lambda \cdot
(Pv_n)(z,t) +
 \frac{\partial}{\partial t}v_{n+1}(z,t) =0 \label{eqn-vn-1}\\
 & &v_{n+1}(z,T)= (\zeta\cdot(z-K_1))^+ ,\label{eqn-vn-2}
\end{eqnarray}
for $(z,t)\in \R\times [0,T]$. The operator $\A(t)$ is defined 
as
 \begin{equation}\label{eq:def-At}
 \A(t) \triangleq  -\mu(q_t-z)\frac{\partial}{\partial z} + \frac12 \sigma^2
 (q_t-z)^2\frac{\partial^2}{\partial z^2}.
 \end{equation}
(iv) The function $v_{\infty}$
is the unique classical solution, i.e. $v_{\infty} \in C^{2,1}$, of
\begin{eqnarray}
& &\A(t)v_{\infty}(z,t) -\lambda\xi v_{\infty}(z,t) + \lambda
\cdot (Pv_{\infty})(z,t) +
 \frac{\partial}{\partial t}v_{\infty}(z,t) =0 \label{eqn-vinf-1}\\
 & &v_{\infty}(z,T)= (\zeta\cdot(z-K_1))^+ .\label{eqn-vinf-2}
\end{eqnarray}
\end{theorem}
\begin{proof}
See Section~\ref {sec:proof}.
\end{proof}

The iterative procedure in (\ref{eqn-vn-1}) simply collapses to the Vecer's PDE (see \cite{vecer}) when $\lambda=0$, i.e. when the underlying asset is a geometric Brownian motion. That is, the iteration in (\ref{eqn-vn-1})  is designed for the models in which the asset price jumps.

\begin{corollary}\label{theorem:main}
Let $V(S_0)$ be as in (\ref{eq:def-asian-price-P}), i.e. $V(S_0)$
is the value of the Asian option for jump diffusion $S$ whose
dynamics is given in (\ref{eq:dyn-S}). Then we have
\begin{equation}\label{eq:rel-V-v}
V(S_0) = S_0 \cdot v_{\infty}(z,0),
\end{equation}
in which $v_{\infty}(\cdot,\cdot)$ is the unique solution of the
integro-partial differential equation (\ref{eqn-vinf-1}) with
terminal condition (\ref{eqn-vinf-2}), and
\begin{equation}\label{eq:defn-z-x0-s0}
z=\frac{X_0}{S_0}=\frac{1}{rT}\left(1-e^{-rT}\right)-e^{-rT}\frac{K_2}{S_0}.
\end{equation}
\end{corollary}

\begin{proof}
Let us define
\begin{equation}\label{eq:def_Mt}
M_t= v_{\infty}(Z^J_t,t), \quad t\in [0,T]
\end{equation}
where $Z^J$, defined in (\ref{eq:def-Zt}), has the initial value $Z^J_0=z$. It follows from
(\ref{eqn-vinf-1}) and the It\^{o}'s lemma that ${M_t}$ is a
$\Q$-martigale, i.e. $M_t=\E^{\Q}\{M_T|\F_t\}$. As a result
\begin{equation}\label{eq:value}
v_{\infty}(z,0) = M_0 = \E^{\Q}_{0,z}\{M_T\} =
\E^{\Q}_{0,z}\{v_{\infty}(Z^J_T, T)\} =
\E^{\Q}_{0,z}\{(\zeta\cdot(Z^J_T-K_1))^+\} = \frac{V(S_0)}{S_0}.
\end{equation}
The last identity follows from the representation
(\ref{eq:def-asian-price-Q}).
\end{proof}

\section{Computing the Prices of Asian Options Numerically}\label{sec:num}

It follows from Theorem~\ref{theorem-2} that the
the sequence of functions $(v_n)_{n\geq 0}$ converges uniformly and exponentially fast to
 $v_{\infty}$ on any compact
domain. Therefore a few iterations of (\ref{eqn-vn-1}) and (\ref{eqn-vn-2}), starting from $v_0$ will produce an accurate approximation to $v_{\infty}$. To perform the iterations we will make use of the finite difference methods for PDEs.
Since each $v_{n+1}(\cdot,\cdot)$ is the classical solution of a partial differential equation (not an integro-partial differential equation) we can use Crank-Nicolson discretization (see page 155 of \cite{wilmott}) along with the SOR algorithm (see e.g. page 150 of \cite{wilmott}) to solve the sparse system of linear equations. In this iteration we will need to compute the integral $P v_n$, and we do this by the trapezoidal rule. 

We will describe this numerical method more precisely in Section~\ref{sec:3.1} and investigate the convergence properties in Section~\ref{sec:convnum}. In Section~\ref{Sec:3-2} we determine the performance (the speed and accuracy characteristics) of our numerical method for the jump diffusion models of \cite{kou} and \cite{merton}. 
In this section we will take the Monte-Carlo simulation results as a benchmark.

\subsection{A numerical algorithm}\label{sec:3.1}
Let us discretize (\ref{eqn-vinf-1}) using the Crank-Nicolson
method. For fixed $\Delta t$, $\Delta z$, $\zM$ and $\zm$, let $M
\Delta t=T$ and $K \Delta z = \zM-\zm$. Let us denote $z_k \triangleq z_{min}+k\Delta z$ $k=0, 1, \cdots, K$. By $\tv$ we will denote the
solution of the difference equation
\begin{equation}\label{eq:tv}
\begin{split}
&(1+p_{k,m}^0)\tv(k,m)-p_{k,m}^+ \tv(k+1,m)-p_{k,m}^-
\tv(k-1,m)\\
&\hspace{2cm}= p_{k,m+1}^+ \tv(k+1,m+1) +p_{k,m+1}^- \tv(k-1,m+1)
+(1-p_{k,m+1}^0)\tv(k,m+1)\\
&\hspace{2.3cm}+\frac{1}{2}\lambda \Delta t
\left[\left(\widetilde{P} \tv\right)(k,m+1)+\left(\widetilde{P}
\tv\right)(k,m)\right],
\end{split}
\end{equation}
for $m=M-1, M-2, \cdot, 0$, $k=0, 1, \cdots, K$,
satisfying the terminal condition $ \tv(k,m)= (\zeta \cdot (k \Delta z-K_1))^+$ and appropriate boundary conditions.  The
coefficients $p_{k,m}^+$, $p_{k,m}^-$ and $p_{k,m}^0$ are given by
\begin{equation}
\begin{split}
p_{k,m}^+&=\frac{1}{4}\left[\sigma^2 \left(\frac{q_{m \Delta t}-k \Delta
z}{ \Delta z}\right)^2-\mu \frac{q_{m \Delta t}-k \Delta z}{ \Delta
z}\right] \Delta t,
\\p_{k,m}^-&=\frac{1}{4}\left[\sigma^2 \left(\frac{q_{m \Delta t}-k \Delta z}{ \Delta z}\right)^2 + \mu \frac{q_{m \Delta t}-k \Delta z}{ \Delta z}\right] \Delta t,
\\p_{k,m}^0&=p_{k,m}^++p_{k,m}^-+\frac{1}{2} \lambda \xi \Delta t.
\end{split}
\end{equation}

In (\ref{eq:tv}), $\tilde{P}$ is a linear operator which is the
discrete version of the operator $P$ in (\ref{eq:op-P}). Letting
$x=\log{y}$, we can write $Pf$ as
\begin{equation}\label{eq:op-P-log}
Pf(z, t) = \int_{\R} f\left(\frac{z}{e^x} + q_t
\frac{e^x-1}{e^x},t\right) e^x F(dx),
\end{equation}
in which $F(dx)$ is the distribution of a random variable $X$,
such that the distribution of $e^X$ is given by the jump measure $\nu$. We
approximate the integral in (\ref{eq:op-P-log}) using trapezoidal
rule. Discretizing a sufficiently large interval $[x_{min},
x_{max}]$ into $L$ subintervals, we obtain the grid $x_{min} =x_0 <x_1
< \cdots < x_L = x_{max}$. This grid may not be equally spaced.
One can choose the grid to be finer where density of the distribution $F$ is
large. The left hand side of (\ref{eq:op-P-log}) will be evaluated
on grid point $(z_k, m\Delta t)$.  However for some $m, k$ and
$\ell$, $z_k / e^{x_{\ell}}+ q_{m \Delta t} (e^{x_{\ell}}-1)/ e^{x_{\ell}}$,
as the first variable of $f$ in the integrand, may not land on
$z_{k'}$ for some $k'$. Consequently, we will determine the value
of the integrand in (\ref{eq:op-P-log}) by linear interpolation.
If
\[
z_{k'} \leq \frac{z_k}{e^{x_{\ell}}} + q_{m \Delta t}
\frac{e^{x_{\ell}}-1}{e^{x_{\ell}}} \leq z_{k'+1}, \quad \text{
for some } k'
\]
then
\[
\tilde{f}\left(\frac{z_k}{e^{x_{\ell}}}+ q_{m \Delta t}
\frac{e^{x_{\ell}}-1}{e^{x_{\ell}}}, m\Delta t\right) =
(1-w)f(z_{k'} , m\Delta t) + w f(z_{k'+1}, m\Delta t) + O((\Delta
z)^2),
\]
in which $w$ is the interpolation weight. On the other hand, if
$z_k / e^{x_{\ell}}+ q_{m \Delta t} (e^{x_{\ell}}-1)/ e^{x_{\ell}}$ is out
side the interval $[z_{min}, z_{max}]$, the value of the function
is determined by the boundary conditions. Now the integral in
(\ref{eq:op-P-log}) can be approximated as

\begin{equation}\label{eq:numint}
\begin{split}
&\int_{\R} f\left( \frac{z}{e^x} + q_t \frac{e^x-1}{e^x},t
\right)e^x F(dx) \\
&= \sum_{\ell=0}^{L-1} \frac12
\left[\tilde{f}\left(\frac{z_k}{e^{x_{\ell}}}+q_{m \Delta t}
\frac{e^{x_{\ell}}-1}{e^{x_{\ell}}}, m\Delta
t\right)e^{x_{\ell}}g(x_{\ell}) +
\tilde{f}\left(\frac{z_k}{e^{x_{\ell+1}}}+q_{m \Delta t}
\frac{e^{x_{\ell+1}}-1}{e^{x_{\ell+1}}}, m\Delta
t\right)e^{x_{\ell+1}}g(x_{\ell+1})\right]\\
& \hspace{1.5cm} \cdot (x_{\ell+1}-x_{\ell}) + O((\Delta x)^2),
\end{split}
\end{equation}
where $\Delta x = \max_{\ell}{(x_{\ell+1}-x_{\ell})}$, and $g$ is the density of $F$.

Note that numerically solving the system of equations in
(\ref{eq:tv}) is quite difficult due to the contributions from the
integral terms (i.e. the $\tP \tv(k,m)$ term). Discretizing 
(\ref{eqn-vn-1}) recursively (using the Crank-Nicolson
discretization) we obtain 
sequence $\tv_n(k,m)$, $n=0,1,
\cdots$, by setting $\tv_0(k,m)=(\zeta \cdot (k \Delta z-K_1))^+$
through the recursive relationship
\begin{equation} \label{eq:tvn}
\begin{split}
&(1+p_{k,m}^0)\tv_{n+1}(k,m)-p_{k,m}^+ \tv_{n+1}(k+1,m)-p_{k,m}^-
\tv_{n+1}(k-1,m)\\
&\hspace{2cm} =p_{k,m+1}^+ \tv_{n+1}(k+1,m+1) + p_{k,m+1}^-
\tv_{n+1}(k-1,m+1)
+(1-p_{k,m+1}^0)\tv_{n+1}(k,m+1)\\
&\hspace{2.3cm}+\frac{1}{2}\lambda \Delta t
\left[\left(\widetilde{P} \tv_n\right)(k,m+1)+\left(\widetilde{P}
\tv_n\right)(k,m)\right].
\end{split}
\end{equation}
For each $n$ the terminal condition $\tv_{n+1}(k,m)= (\zeta \cdot
(k \Delta z-K_1))^+$ and appropriate boundary conditions are
satisfied. We will solve the sparse linear system of equations using the SOR method (see e.g. \cite{wilmott}).

\subsection{Convergence of the Numerical Algorithm} \label{sec:convnum}
In what follows we will first show that as $n \rightarrow \infty$,
$\tv_n$ converges to $\tv$. Next, we will argue that as the mesh
sizes $\Delta t$ and $\Delta z$ go to zero $\tv$ converges to
$v_{\infty}$. For the sake of simplicity of the presentation, in
what follows we will assume that $(\tP 1) (k,m) \leq \int_{R_+} y
\nu(dy) = \xi$. Otherwise the order of error of the discretization
of the integral will have to be sufficiently small for the
following statement to be true.

\begin{proposition}\label{prop:disc-iter-conv}
For sufficiently  small $\Delta t$ and $\Delta z$, $\tilde{v}_n$
converges to $\tv$ uniformly and at an exponential rate.
\end{proposition}
\begin{proof}
Let us define
\begin{equation}
e_{n}(k,m) \triangleq \tv(k,m)-\tv_n(k,m).
\end{equation}
Since $\tP$ is a linear operator $e_n$ will satisfy (\ref{eq:tvn})
when we replace $\tv_n$ by $e_n$ and $\tv_{n+1}$ by $e_{n+1}$. Now
let us define
\begin{equation}
E_n^m=\max_{k}|e_n(k,m)|, \quad E_n=\max_{m,k}|e_n(k,m)|.
\end{equation}
Choosing $\Delta t$ and $\Delta z$  small enough we can guarantee
that $p^+_{k,m}$, $p^-_{k,m}$ and $1-p^0_{k,m}$ are positive for all $(k,m)$. As a result it
follows from the difference equation $e_n$ satisfies that
\begin{equation}\label{eq:err}
\begin{split}
(1+p_{k,m}^0)\left|e_{n+1}(k,m)\right|  &\leq (p_{k,m}^++p_{k,m}^-) E_{n+1}^m +
(p_{k,m+1}^++p_{k,m+1}^-+(1-p_{k,m+1}^0))E_{n+1}^{m+1}+\lambda
\Delta t \xi E_n \\
& = (p_{k,m}^++p_{k,m}^-) E_{n+1}^m +\left(1-\frac12 \lambda \xi \Delta t
\right) E_{n+1}^{m+1} + \lambda \Delta t \xi E_n,
\end{split}
\end{equation}
in which we used the assumption that $(\tP 1) (k,m) \leq \xi$. It follows from (\ref{eq:err}) that 
\begin{equation}\label{eq:err2}
\begin{split}
  (1+p_{k,m}^0)\left|e_{n+1}(k,m)\right|  - \left(p^0_{k,m} - \frac12 \lambda\xi\Delta t \right) E_{n+1}^m \leq \left(1-\frac12 \lambda \xi \Delta t
\right) E_{n+1}^{m+1} + \lambda \Delta t \xi E_n.
\end{split}
\end{equation}
Let $k^*$ be such that $|e_{n+1}(k^*,m)| = E_{n+1}^m$.
Since the right-hand-side of (\ref{eq:err2}) does not depend on
$k$, we can take $k=k^*$ on the left-hand-side to write
\begin{equation}\label{eq:recur-ineq}
 E_{n+1}^m \leq \theta E_{n+1}^{m+1}+ (1-\theta) E_{n},
\end{equation}
in which
\begin{equation}
\theta \triangleq \frac{1-1/2 \cdot \lambda \xi \Delta t}{1+1/2
\cdot \lambda \xi \Delta t} \in (0,1],
\end{equation}
as a result of the assumption $p^+_{k,m}$, $p^-_{k,m}$ and $1-p^0_{k,m}$ are
positive for all $(k,m)$. It follows from (\ref{eq:recur-ineq}) that
\begin{equation}\label{eq:Err}
E_{n+1}^m \leq \theta^{M-m}E_{n+1}^M+(1-\theta)(1+\theta+\cdots+\theta^{M-m-1})E_n.
\end{equation}
Because of the terminal condition of $\tilde{v}_n$, we have
$E_{n+1}^M=0$. In addition, (\ref{eq:Err}) is satisfied for all
$m$ we get that
\begin{equation}\label{eq:disc-err-conv-rate}
E_{n+1} \leq (1-\theta^M)E_n.
\end{equation}
As a result
\begin{equation}\label{eq:conv-rate}
E_n \leq (1-\theta^M)^n E_0 \rightarrow 0 \quad \text{as $n \rightarrow \infty$}.
\end{equation}
\end{proof}
\begin{remark}\label{eq:etalim}
As $M \rightarrow \infty$
\begin{equation}\label{eq:lim-theta}
1- \theta^M=1-\left(\frac{1-1/2\cdot \lambda \xi \cdot T/M}{1+1/2\cdot
\lambda \xi \cdot T/M}\right)^M \rightarrow 1-e^{-\lambda \xi T},
\end{equation}
which shows that the convergence rate in (\ref{eq:conv-rate}) agrees with the convergence rate in (\ref{eq:uniform-conv}).
\end{remark}

\begin{proposition}\label{prop:disc-conv}
\begin{equation}
|v_{\infty}(z_k, m \Delta t)-\tv(k,m)| \rightarrow 0
\end{equation}
as $\Delta z, \Delta t, \Delta x \rightarrow 0$.
\end{proposition}
\begin{proof}
Using the triangle inequality let us write
\begin{equation}\label{eq:err-bdd}
\begin{split}
|v_{\infty}(z_k, m \Delta t)-\tv(k,m)| &\leq |v_{\infty}(z_k, m \Delta t)-v_n(z_k, m \Delta t)|+|v_n(z_k, m \Delta t)-\tv_n(k,m)|\\& \quad +|\tv_n(k,m)-\tv(k,m)|
\\ &\leq  C \left(1-e^{-\lambda \eta(T-m\Delta t)}\right)^n+ n \cdot O((\Delta t)^2+(\Delta z)^2+(\Delta x)^2)+\widetilde{C} (1-\theta^M)^n,
\end{split}
\end{equation}
for some positive constants $C$ and $\widetilde{C}$. The first and the third terms in the right-hand-side of the second inequality are due to (\ref{eq:uniform-conv}) and (\ref{eq:disc-err-conv-rate}). The second term arises since the order of error from discretizing a PDE using a Crank-Nicolson scheme is $O((\Delta z)^2+(\Delta t)^2)$, the interpolation and the discretization error from the numerical integration are of order $(\Delta z)^2$ and $(\Delta x)^2$ and that the total error made at each step propagates at most linearly in $n$ when we sequentially discretize the PDEs in (\ref{eqn-vn-1}).

Letting $\Delta t, \Delta z \rightarrow 0$ in (\ref{eq:err-bdd}) we obtain that
\begin{equation}
\lim_{\Delta t, \Delta z \rightarrow 0} |v_{\infty}(z_k, m \Delta t)-\tv(k,m)| \leq C \left(1-e^{-\lambda \eta T}\right)^n+ \widetilde{C} \left(1-e^{-\lambda \xi T}\right)^n,
\end{equation}
in which we used (\ref{eq:lim-theta}). Since $n$ is arbitrary the
result follows.
\end{proof}

\begin{remark}
The proof of Proposition~\ref{prop:disc-conv} leads itself to an
order of convergence.  In (\ref{eq:err-bdd}) we can choose
$n=O(\log(1/[(\Delta t)^2+(\Delta z)^2+(\Delta x)^2]))$, which would guarantee
that the first and the third terms on the right-hand-side are of
order $O((\Delta t)^2+(\Delta z)^2+(\Delta x)^2)$, and the right-hand-side
becomes of order $[(\Delta t)^2+(\Delta z)^2+(\Delta x)^2] \log(1/[(\Delta
t)^2+(\Delta z)^2+(\Delta x)^2])$. Note that this order of convergence is
better than that of $O((\Delta t)^{2-\gamma}+(\Delta
z)^{2-\gamma}+(\Delta x)^{2-\gamma})$ for any $\gamma>0$.
\end{remark}

\subsection{Numerical results for Kou's and Merton's models}\label{Sec:3-2}
There are two well-known examples of jump diffusion in the
literature, the double exponential model as in \cite{kou} and the
normal model as in \cite{merton}. In this section, we will
demonstrate our algorithm in Section \ref{sec:3.1} in pricing
Asian options for these two models. We will introduce the jump
distributions chosen by \cite{kou} and \cite{merton} next. Let $X$
be a random variable whose probability distribution function is
equal to a given distribution $F$ and let the jump measure $\nu$
be equal to the distribution of the random variable $e^{X}$. In
Kou's model, $F$ is the double exponential distribution whose
density is
\begin{equation}\label{eq:den-douexp}
F(dx) = \left(p\,\eta_1 e^{-\eta_1 x}1_{\{x\geq 0\}} +
(1-p)\,\eta_2 e^{\eta_2 x} 1_{\{x<0\}}\right) dx.
\end{equation}
In Merton's model, $F$ is the normal distribution whose density is
\begin{equation}\label{eq:den-normal}
F(dx) = \frac{1}{\sqrt{2\pi \tilde{\sigma}^2}}
\exp\left(\frac{-(x-\mu)^2}{\tilde{\sigma}^2}\right)dx.
\end{equation}

The price of the Asian option with floating strike $K_1$ and fixed
strike $K_2$, whose pay-off function is given by (\ref
{eq:def-asian-price-P}), can be calculated in terms of
$v_{\infty}(z,0)$ as a result of Corollary~\ref{theorem:main}. On
the other hand, in Proposition \ref{prop:disc-conv} we have shown
that $v_{\infty}$ can be approximated by its
discrete version $\tilde{v}$. Furthermore, Proposition
\ref{prop:disc-iter-conv} tells us that $\tilde{v}$ can be
approximated by the sequence $(\tilde{v}_n)_{n\geq 0}$ with
exponentially fast convergence. Therefore, we will approximate the
price of Asian options by iteratively solving a sequence of
difference equations in (\ref{eq:tvn}), which are discretization
of a sequence of parabolic partial differential equations (not
integro-differential equations) given by (\ref{eqn-vn-1}) and
(\ref{eqn-vn-2}).

In the following, we will list the numerical results for the
prices of Asian options. Since we could not find any numerical
results on European Asian options for jump diffusion models in the
literature, we use the Put-Call parity for Asian options as a
consistency check for our results. For the European Asian option
with floating strike $K_1$ and fixed strike $K_2$, the put - call
parity gives the following identity between call and put option
price
\begin{equation}\label{eq:p-cparity}
C(S_0, 0) -P(S_0, 0) = \frac{1}{rT}(1-e^{-rT})S_0 - K_1 S_0 -
e^{-rT}K_2.
\end{equation}
This identity does not depend on dynamics of the underlying
process $S$. Using our algorithm, we will calculate the call and
put option price independently. Then we will compare the
difference between our call and put price with the difference
predicted by (\ref{eq:p-cparity}). In addition, we will also
compare our numerical results with the Monte Carlo results.

In Tables 1 and 2, we will list the numerical results for the
prices of European Asian options for both Kou's model and Merton's
model. Run times are in seconds. All our computations are
performed on a Pentium IV 3.0 GHz machine with C++ implementation.
In Tables 3, 4 and 5, we will list the convergence results of the
for the Asian option prices for the double exponential jump model.
The parameters for both call and put options are the same as the
7th row in Table 1, i.e. $r=0.15$, $S_0=100$, $K_1=0$, $K_2=90$,
$T=1$, $\sigma=0.2$, $\lambda=1$ and $\eta_1=\eta_2=25$.

As we can see from these tables, our algorithm is stable with
respect to all parameters and the convergence is fast. Moreover,
our difference between call and put option prices are within $\pm
0.01$ comparing to the difference predicted by put-call parity in
(\ref{eq:p-cparity}). The call option prices are almost within the
standard error of the Monte Carlo results.

If the dynamics of $Z^J$ only contains the diffusion part (i.e.
$\lambda=0$), our algorithm is simply SOR. Using the same
parameters chosen in \cite{vecer}, the SOR gives approximate
option price with error $\pm1\times 10^{-3}$ dollar (we compared
our results to Table 2 in \cite{vecer}), and the run times are
below 0.02 second. Comparing the numerical results in Tables 1 and
2, we see that the evaluation of the integral term by numerical
integral is the time consuming part. This can be speeded up by
using the Fast Fourier Transform (see e.g. \cite{A-O} for the
application of the Fast Fourier Transform in American options
under the Variance Gamma model).\section{Mathematical Analysis}

The purpose of this section is to provide the necessary background to prove Theorem~\ref{theorem-2}.
First, in Section~\ref{sec:propJ} we will analyze the properties of functional operator $J$: We will analyze how $J$ increases the regularity of certain class of functions. $J$ takes functions with certain regularity properties into the unique solutions of parabolic differential equations and gives them more regularity (see Lemmas~\ref{lemma:lip}-\ref{lemma:holder} and Proposition~\ref{theorem-1}). 

Next, in Section~\ref{sec:analyze-vn}, we will develop the properties of the functions defined in (\ref{eq:def-vn}) in a sequence of lemmas and corollaries using the results developed in Section~\ref{sec:propJ}. These properties will then be used  to prove Theorem~\ref{theorem-2} in Section~\ref{sec:proof}.

\subsection{Properties of Operator $J$}\label{sec:propJ}
First, we will develop a representation of the functional operator $J$ that is amenable to regularity analysis, which is carried out in this section.
Using the notation in page 8 of \cite{pham1}, we can rewrite $J$ as
\begin{equation}\label{eq:def-J-exp}
Jf(z,t)= \E^{\Q} \left\{ e^{-\lambda\xi (T-t)} \left(\zeta \cdot \left(Z_{T-t}^{t,z} - K_1\right)\right)^+ + \int_0^{T-t}e^{-\lambda\xi s}\lambda \cdot Pf\left(Z_s^{t,z}, t+s\right) ds\right\},
\end{equation}
in which
\begin{equation}\label{eq:def-pf}
Pf\left(Z_s^{t,z} , t+s\right) = \int_{\R^+} f\left(\frac{Z_s^{t,z}}{y} + q_{t+s} \frac{y-1}{y}, t+s\right) y \nu(dy).
\end{equation}
the process $Z^{t,z} = \{Z_{s}^{t,z}; s \geq 0\}$
has the dynamics
\begin{equation}\label{eq:z-t}
dZ_{s}^{t,z} = -\mu (q_{t+s}-Z_{s}^{t,z})ds +
\sigma(q_{t+s}-Z_{s}^{t,z})dW_{s}, \quad Z_0^{t,z}=z,
\end{equation}
where $\{W_s\}_{s\geq 0}$ is a Wiener process under the measure $\Q$.
It is possible to determine the solution to (\ref{eq:z-t}) explicitly. For this purpose it will be convenient
to work with the process $\tilde{Z}_s \triangleq q_{t+s}-Z_s^{t,z}$. It follows from
(\ref{eq:z-t}) that the dynamics of $\tilde{Z}$ are given by
\begin{equation}\label{eq:tilde-z-t}
d\tilde{Z}_{s} = \mu \tilde{Z}_{s} ds - \sigma \tilde{Z}_{s}
dW_{s} + g(t+s) ds, \quad \tilde{Z}_0=\tilde{z}=q_t-z,
\end{equation}
in which $g(t)= \frac{d}{dt}q_t$. Now it is easy to obtain the solution of stochastic
differential equation (\ref{eq:tilde-z-t}) as
\begin{equation}\label{sol tilde-z-t}
\tilde{Z}_{s} = \tilde{z}H_s^0 + \int_0^{s} H_{s-v}^0
\,g(t+v) dv \quad \text{for } s \geq 0,
\end{equation}
in which
\begin{equation}
H_s^0 \triangleq \exp((\mu-\frac12 \sigma^2)s - \sigma W_s)
\label{eq:def-h}.
\end{equation}
As a result we have that
the solution of (\ref{eq:z-t}) is given by
\begin{equation}\label{sol z-t}
Z_{s}^{t,z} = z H_s^0 + b_s, \quad s \geq 0,
\end{equation}
in which
\begin{equation}
b_s \triangleq q_{t+s} -q_t H_s^0 - \int_0^{s} H_{s-v}^0
\,g(t+v) dv \label{eq:def-b}.
\end{equation}
It follows from (\ref{sol z-t}) that the solution of
the stochastic differential equation (\ref{eq:z-t}) is linear with
respect to its initial value $z$. Inserting its solution (\ref{sol
z-t}) back into the definition of the operator $J$ in
(\ref{eq:def-J-exp}), we obtain
\begin{equation}\label{eq:def-j-z}
Jf(z,t) =
\E^{\Q}\left\{e^{-\lambda\xi(T-t)}\left(\zeta\cdot\left(zH_{T-t}^0
+ b_{T-t}-K_1\right)\right)^+ + \int_0^{T-t}
e^{-\lambda\xi s}\lambda\cdot
Pf(zH_s^0+b_s,t+s)\,ds\right\}.
\end{equation}

In the following, we will study the regularity properties of the
operator $J$ with respect to both space and time. When the
function $f$ is Lipschitz continuous with respect to its first
variable, the following lemmas show $Jf$ is not only Lipschitz
with respect to its first variable, but also H\"{o}lder continuous
with respect to the second variable.

\begin{lemma}\label{lemma:lip}
For any $t \in [0,T]$, let us assume the function $f$ satisfies
\begin{equation}\label{eq:lip-f}
|f(z,t)-f(\tilde{z},t)|\leq D\,|z-\tilde{z}|, \quad  z,\tilde{z}
\in \R,
\end{equation}
for a positive constant $D$ that only depends on $T$.
Then $Jf$ satisfies
\begin{equation}\label{eq:lip-jf}
|Jf(z,t)-Jf(\tilde{z},t)|\leq E\, |z-\tilde{z}|, \quad z,\tilde{z}
\in \R,
\end{equation}
with $E=\max\{1,D\}$.
\end{lemma}
\begin{proof}
>From the definition of operator $J$ in (\ref{eq:def-j-z}), we have
\begin{eqnarray}
|Jf(z,t)- Jf(\tilde{z},t)|  &\leq&
\E^{\Q}\bigg\{e^{-\lambda\xi(T-t)}\left|\left(\zeta\cdot\left(zH_{T-t}^0+b_{T-t}-K_1\right)\right)^+
-\left(\zeta\cdot\left(\tilde{z}H_{T-t}^0 + b_{T-t}-K_1\right)\right)^+\right|  \nonumber\\
& &  + \int_0^{T-t} ds \,e^{-\lambda\xi s}\lambda \cdot
\left|Pf(zH_{s}^0 + b_{s}, t+s)
-Pf(\tilde{z}H_{s}^0+b_{s},t+s)\right| \bigg\}. \label{eq:x-tx}
\end{eqnarray}
Let us obtain a bound on the right hand side of (\ref{eq:x-tx}). First observe that
\begin{equation}
\left|(\zeta\cdot(zH_{T-t}^0+b_{T-t}-K_1))^+ -
(\zeta\cdot(\tilde{z}H_{T-t}^0+b_{T-t}-K_1))^+\right| \leq
|z-\tilde{z}|H_{T-t}^0 ,\label{est1}
\end{equation}
and
\begin{eqnarray}
\big|Pf(zH_{s}^0+ b_{s},t+s) &-& Pf(\tilde{z}H_{s}^0 +
b_{s},t+s)\big|
\nonumber\\
&\leq& \int_{\R_+}\left|f\left(\frac{zH_{s}^0}{y} +
\frac{b_{s}}{y} + q_{t+s}
\frac{y-1}{y},t+s\right)-f\left(\frac{\tilde{z}H_{s}^0}{y} +
\frac{b_{s}}{y} + q_{t+s} \frac{y-1}{y},t+s\right)\right| y \nu(dy)\nonumber\\
&\leq& \int_{\R_+} D\, |zH_{s}^0 - \tilde{z}H_{s}^0| \,
\nu(dy)
\nonumber\\
&=& D\, |z-\tilde{z}| \, H_{s}^0 \label{est2}.
\end{eqnarray}
On the other hand, from the definition of $H_s^0$ in (\ref{eq:def-h}), we have that
\begin{equation}\label{eq:eh}
\E^{\Q}\{H_{s}^0\} = e^{\mu s}.
\end{equation}
Inserting (\ref{est1}), (\ref{est2}) and (\ref{eq:eh}) back into
the equation (\ref{eq:x-tx}), we have
\begin{eqnarray}
\left|Jf(z,t)- Jf(\tilde{z},t)\right| &\leq&
e^{-\lambda\xi(T-t)}\, |z-\tilde{z}| \, \E^{\Q}\{H_{T-t}^0\} +
\int_0^{T-t} ds\, e^{-\lambda\xi s}\, \lambda \, D \, |z-\tilde{z}|
\,
\E^{\Q}\{H_{s}^0\} \nonumber\\
&=& |z-\tilde{z}|\, \left(e^{(\mu-\lambda\xi)(T-t)} + D\int_0^{T-t} ds
\, \lambda \, e^{(\mu-\lambda\xi)s}\right) \nonumber\\
&\leq& \left(D+(1-D)e^{-\lambda (T-t)}\right)\,|z-\tilde{z}|\\
&\leq& \max\{1,D\}\, |z-\tilde{z}|.
\end{eqnarray}
\end{proof}
\begin{remark}\label{remark-2.2}
Let us define
\begin{eqnarray}
 M_f &\triangleq& \sup_{t\in[0,T]} f(0,t),\label{eq:def-Mf}\\
 M_{Jf} &\triangleq& \sup_{t\in[0,T]} Jf(0,t). \label{eq:def-MJf}
\end{eqnarray}
It follows from the Lipschitz conditions (\ref{eq:lip-f}) and
(\ref{eq:lip-jf}) that both $f$ and $Jf$ satisfy
linear growth conditions, if $M_f$ and $M_{Jf}$ are finite, since for $(z,t)\in \R\times[0,T]$
\begin{eqnarray}
f(z,t) \leq f(0,t) + D\,|z|,\label{eq:linear-f}\\
Jf(z,t) \leq Jf(0,t) + E\,|z|.\label{eq:linear-jf}
\end{eqnarray}
\end{remark}

In the next two lemmas we will need the following moment estimates  of $Z_s^{t,z}$.
\begin{eqnarray}
 \E^{\Q} \left\{|Z_s^{t,z}|\right\} &\leq& C(1+|z|), \label{eq:eZ}\\
 \E^{\Q}\left\{|Z_s^{t,z} - z| \right\} &\leq&
 C(1+|z|)s^{\frac12}, \label{eq:eZ-z}
\end{eqnarray}
in which $0\leq s \leq T$ and $C$ is a constant depending on
$T$. These estimates can be  found  in
\cite{pham1} (Lemma 3.1).

\begin{lemma}\label{lemma:mjf}
We have that
\begin{equation}\label{est:Mjf}
M_{Jf} \leq U + \alpha\left(M_f + \frac{B}{\xi}\right),
\end{equation}
in which $\alpha=1-e^{-\lambda\xi T} <1$, and $U$, $B$ are
positive constants depending on $T$.
\end{lemma}
\begin{proof}
We will estimate $M_{Jf}$ using the definition of the operator $J$ in
(\ref{eq:def-J-exp}). First, we have that
\begin{equation*}
\begin{split}
\E^{\Q}
\left\{e^{-\lambda\xi(T-t)}\left(\zeta\cdot\left(Z_{T-t}^{t,0} -
K_1\right)\right)^+\right\}  &\leq \E^{\Q}
\left\{e^{-\lambda\xi(T-t)}\left(\left|Z_{T-t}^{t,0}\right|+K_1\right)\right\}\\
&= \E^{\Q}
\left\{e^{-\lambda\xi(T-t)}\left(\left|b_{T-t}\right|+K_1\right)\right\},
\end{split}
\end{equation*}
in which we obtain the last inequality using the expression of
$Z^{t,z}$ in (\ref{sol z-t}) with $z=0$. First, it follows from
(\ref{eq:eZ}) with $z=0$ that
\begin{equation}\label{eq:est-b}
\E^{\Q}\{|b_{T-t}|\}= \E^{\Q}\{|Z_{T-t}^{t,0}|\}\leq C.
\end{equation}
Letting $U \triangleq C+K_1$,
which is a finite positive constant depending on $T$, we have that
\begin{equation}\label{est-Jf1}
\E^{\Q}\left\{e^{-\lambda\xi(T-t)}\left(\left|b_{T-t}\right|+K_1\right)\right\}
\leq U.
\end{equation}

Second, we will estimate the second term in the definition of
$J$ in (\ref{eq:def-J-exp}). From the definition of $Pf$ in
(\ref{eq:def-pf}), we have
\begin{eqnarray}
\E^{\Q}\left\{Pf(Z_s^{t,0},t+s)\right\} &=&
\E^{\Q} \left\{\int_{\R_+}f\left(\frac{Z_s^{t,0}}{y}+
q_{t+s}\frac{y-1}{y}, t+s\right)\, y\nu(dy)\right\}\nonumber\\
&\leq& \E^{\Q} \left\{\int_{\R_+}\left(f(0,t+s)+ D
\frac{\left|Z_s^{t,0}\right|}{y} + D \,q_{t+s}
\frac{|y-1|}{y}\right)\, y \nu(dy)\right\} \nonumber\\
&\leq& \xi f(0,t+s) + D \, q_{t+s} (\xi+1) + D\, \E^{\Q}\left\{\left|Z_s^{t,0}\right|\right\} \nonumber\\
&\leq& \xi f(0,t+s) + D\, q_{t+s}(\xi+1) + C\cdot D. \label{eq:expPf}
\end{eqnarray}
To obtain the first inequality we use the inequality
(\ref{eq:linear-f}), whereas the second inequality follows from $|y-1|\leq
y+1$. To obtain the last inequality, we use the inequality (\ref{eq:eZ})
with $z=0$. Now, using (\ref{eq:expPf}) we obtain
\begin{equation}\label{eq:lhs-int-pf}
\E^{\Q} \left\{\int_0^{T-t} e^{-\lambda\xi s}\lambda\cdot
Pf\left(Z_s^{t,0},t+s\right)\, ds\right\} \leq
\int_0^{T-t}e^{-\lambda\xi s}\lambda\cdot \left[\xi f(0,t+s) + D\,
q_{t+s} (\xi+1) + C\cdot D\right] \, ds.
\end{equation}
Since $0\leq s\leq T-t$, we have $q_{t+s} \leq \frac{1}{rT}$. Let us define
\begin{equation}\label{eq:def-B}
B\triangleq \left[\frac{1}{rT}(\xi+1)+ C\right]\cdot D,
\end{equation}
which is a finite positive constant depending on $T$. Now, we have the
following estimation on the left hand side of (\ref{eq:lhs-int-pf})
\begin{eqnarray}
 \E^{\Q} \left\{\int_0^{T-t}e^{-\lambda\xi s }\lambda\cdot Pf(Z_s^{t,0}, t+s)\,
 ds\right\}
 &\leq& \int_0^{T-t}e^{-\lambda\xi s}\lambda\cdot
 \left[\xi f(0,t+s) + B\right]\, ds \nonumber\\
 &\leq& \left(1-e^{-\lambda\xi(T-t)}\right)\left(M_f +
 \frac{B}{\xi}\right) \nonumber\\
 &\leq& \left(1-e^{-\lambda\xi T}\right)\left(M_f+
 \frac{B}{\xi}\right), \quad \text{for } t\in [0,T].\label{est-Jf2}
\end{eqnarray}
>From inequalities (\ref{est-Jf1}) and (\ref{est-Jf2}), we
conclude that
\begin{equation}\label{est-J0t}
Jf(0,t)\leq U+ \left(1-e^{-\lambda\xi T}\right)\left(M_f+
\frac{B}{\xi}\right).
\end{equation}
\end{proof}

\begin{remark}\label{remark-tildeDE}
Lemma \ref{lemma:mjf} and Remark \ref{remark-2.2} indicate that
\begin{eqnarray}
 f(z,t) &\leq& M_f + D\,|z| \leq \tilde{D}(1+|z|), \label{eq:tildeD} \\
 Jf(z,t) &\leq&  U+ \alpha\left(M_f+\frac{B}{\xi}\right) + E\, |z|
 \leq \tilde{E}(1+|z|), \label{eq:tildeE}
\end{eqnarray}
in which $\tilde{D}=\max\{M_f, D\}$ and
$\tilde{E}=\max\{U+\alpha(M_f+B/\xi), E\}$. We will use these
linear growth properties to show a regularity property of the
operator $J$ with respect to time in the next lemma.
\end{remark}

\begin{lemma}\label{lemma:holder}
Assume the function $z \mapsto f(z,t)$ satisfies
\begin{equation}\label{eq:lip-f2}
 |f(z,t)-f(\tilde{z},t)| \leq D\, |z-\tilde{z}|,
\end{equation}
for $z,\tilde{z}\in \R$ as in Lemma \ref{lemma:lip} and $M_f<\infty$. Then $t
\mapsto Jf(z,t)$  satisfies
\begin{equation}\label{eq:holder-jf}
 |Jf(z,t)-Jf(z,s)| \leq F\, (1+|z|)\,(s-t)^{\frac12}, \quad 0\leq t < s
 \leq T,
\end{equation}
in which $F$ is a positive constant that only depends on $\lambda$, $\xi$, $T$
and $M_f$.
\end{lemma}
\begin{proof}
For any $h\in [t, T]$, it follows from the definition of operator $J$ in
(\ref{eq:def-J-exp}) and the Markov property of $Z_s^{t,z}$ that
\begin{equation}
Jf(z,t) = \E^{\Q} \left\{\int_0^{h-t} dv\,
e^{-\lambda\xi v}\,\lambda \cdot Pf(Z_v^{t,z},t+v) +
e^{-\lambda\xi(h-t)}Jf(Z_{h-t}^{t,z},h) \right\}. \label{eq:dyn}
\end{equation}
With $h=s$,
\begin{eqnarray}
&&\left|Jf(z,t)-Jf(z,s)\right| \leq \E^{\Q}\left\{\int_0^{s-t}
e^{-\lambda\xi v}\lambda \cdot Pf(Z_v^{t,z},t+v)\,dv +
|e^{-\lambda\xi(s-t)}Jf(Z_{s-t}^{t,z},s)-Jf(z,s)|\right\}\nonumber\\
&& \hspace{0.5cm} \leq \E^{\Q} \left\{\int_0^{s-t}
e^{-\lambda\xi v}\lambda\cdot Pf(Z_v^{t,z},t+v)\,dv +
e^{-\lambda\xi(s-t)}\,\left|Jf(Z_{s-t}^{t,z},s)-Jf(z,s)\right|+
\left|e^{-\lambda\xi(s-t)}-1\right|Jf(z,s)\right\}.\nonumber\\
\label{eq:Jft-s}
\end{eqnarray}
In what follows we will bound the terms on the right-hand-side of this inequality.
Since the condition (\ref{eq:lip-f2}) holds, Lemma \ref{lemma:lip}
applies. As a result it follows from (\ref{eq:lip-jf}) that
\begin{equation}\label{est-holder1}
{\E}^{\Q} \left\{|Jf(Z_{s-t}^{t,z},s)-Jf(z,s)|\right\} \leq E\,
{\E}^{\Q} \left\{|Z_{s-t}^{t,z}-z|\right\},
\end{equation}
Using the estimate in (\ref{eq:tildeD}) we have that
\begin{eqnarray}
{\E}^{\Q} \left\{Pf(Z_v^{t,z},t+v)\right\} &=& \int_{\R_+} y\nu(dy)\, \E^{\Q}  \left\{f\left(\frac{Z_v^{t,z}}{y} + q_{t+v}\frac{y-1}{y}, t+v\right)\right\}\nonumber\\
&\leq& \int_{\R_+}y\nu(dy)\,\tilde{D} \left(1+
\frac{1}{y} {\E}^{\Q} \left\{|Z_v^{t,z}|\right\} + q_{t+v} \frac{|y-1|}{y}\right) \nonumber \\
&\leq& \tilde{D}\left(\xi + {\E}^{\Q}\left\{|Z_v^{t,z}|\right\}
+
(\xi+1) q_{t+v}\right)\nonumber\\
&\leq& \tilde{D}\left(\xi + \frac{1}{rT}(\xi+1) +
C(1+|z|)\right).\label{est-holder2}
\end{eqnarray}
To obtain the last inequality we use the estimation (\ref{eq:eZ}) and
the fact that $q_{t+v} \leq \frac{1}{rT}$ for $v\in [0,s-t]$. On the other hand, from
(\ref{eq:tildeE}), we have that
\begin{equation}
|Jf(z,s)| \leq \tilde{E}(1+|z|). \label{est-holder3}
\end{equation}
In the inequalities above, the constants $E, \tilde{D} \text{ and }
\tilde{E}$ are as in Lemma \ref{lemma:lip} and Remark
\ref{remark-tildeDE}.

Now, using (\ref{est-holder1}), (\ref{est-holder2}),
(\ref{est-holder3})
and the inequalities
\begin{equation}\label{est-holder4}
e^{-\lambda\xi v} <1, \quad \text{and} \quad
1-e^{-\lambda\xi(s-t)} \leq \lambda\xi(s-t),
\end{equation}
we can bound
(\ref{eq:Jft-s}) as follows:
\begin{eqnarray}
|Jf(z,t)&-&Jf(z,s)|\nonumber\\
&\leq& \tilde{D}\,\lambda\, \left(\xi + \frac{1}{rT}(\xi+1) +
C(1+|z|)\right)\,(s-t) +
E\,{\E}^{\Q} \left\{|Z_{s-t}^{t,z}-z|\right\} + \lambda\,\xi\,\tilde{E}\,\left(1+|z|\right)\,(s-t) \nonumber\\
&\leq& \tilde{D}\, \lambda\, \left(\xi + \frac{1}{rT}(\xi+1) +
C(1+ |z|)\right)\,(s-t) + E\cdot C \left(1+|z|\right)\,
(s-t)^{\frac12} + \lambda\, \xi\, \tilde{E}\,
\left(1+|z|\right)\, (s-t) \nonumber\\
&\leq& F \left(1+|z|\right)\,(s-t)^{\frac12}, \label{eq:Jft-s-res}
\end{eqnarray}
where $F$ is a positive constant only depending on $\lambda$, $\xi$, $T$
and $M_f$. To obtain the second inequality, we use the moment estimates
(\ref{eq:eZ-z}). To obtain the last inequality, we use the fact that $s-t
\leq T$.
\end{proof}

In the following
proposition we show that $Jf$ satisfies a parabolic partial differential
equation.

\begin{proposition}\label{theorem-1}
 Assume function $f: \R\times [0,T] \rightarrow \R_+$ satisfies
 the following condition
 \begin{equation}\label{theorem-1-condition}
 \left|f(z,t)-f(\tilde{z},s)\right| \leq D|z-\tilde{z}| +
 F(1+|z|)\, |s-t|^{\frac12}, \quad (z,t), (\tilde{z},s) \in
 \R\times[0,T],
 \end{equation}
 in which $D$ and $F$ are constants, then the function $Jf:\R\times [0,T]\rightarrow
 \R_+$ is the unique classical solution, i.e. $Jf \in C^{2,1}$, of
 \begin{eqnarray}
 & &\A(t)Jf(z,t) -\lambda\xi Jf(z,t) + \lambda \cdot Pf(z,t) +
 \frac{\partial}{\partial t}Jf(z,t) =0 \label{eqn-Jf-1}\\
 & &Jf(z,T)=(\zeta\cdot (z-K_1))^+.\label{eqn-Jf-2}
 \end{eqnarray}
 \end{proposition}
\begin{proof}
It is clear that $Jf$ satisfies the terminal condition. For any
point $(z,t)\in\R\times[0,T]$, let us take a rectangle
$R=[z_1,z_2]\times[0,T]$, so that $(z,t)\in R$. Denote the
parabolic boundary of $R$ by $\partial_0 R := \partial R-
[z_1,z_2]\times\{0\}$. Consider the following parabolic partial
differential equation
\begin{eqnarray}
& &\A(t)u(z,t) -\lambda\xi u(z,t) + \lambda \cdot Pf(z,t) +
 \frac{\partial}{\partial t}u(z,t) =0 \label{eqn-Jf-1-pf}\\
& &u(z,t)= Jf(z,t) , \quad \text{on } \partial_0 R.
\label{eqn-Jf-2-pf}
\end{eqnarray}
Because of the condition (\ref{theorem-1-condition}),
$z \rightarrow f(z,t)$ is Lipschitz in its first variable uniformly in
the second variable, it follows from Lemmas \ref{lemma:lip} and
\ref{lemma:holder} that $z \rightarrow Jf(z,t)$ is Lipschitz and $t \rightarrow Jf(z,t)$ is
H\"{o}lder continuous. As a result $Jf(\cdot,\cdot)$ is a
continuous function on $\R \times \R_+$.

On the other hand, for $(z,t), (\tilde{z},s) \in R$, it follows
from the condition (\ref{theorem-1-condition}) that
\begin{eqnarray}
|Pf(z,t)-Pf(\tilde{z},s)| &\leq& \int_{\R_+}
\left|f\left(\frac{z}{y} + q_t \frac{y-1}{y}, t\right)- f\left(\frac{\tilde{z}}{y}+ q_s\frac{y-1}{y},s\right)\right| y\nu(dy)\nonumber\\
&\leq& \int_{\R_+} \left[D|z-\tilde{z}|+ D|q_t-q_s|\,|y-1| +
F\left(y+|z|+ q_t|y-1|\right)\, |s-t|^{\frac12}\right]\nu(dy) \nonumber\\
&\leq& D|z-\tilde{z}| + D (\xi+1) \frac{e^{-rT}}{T} \left|\int_t^s
e^{ru} du\right| + F\big(\xi+q_t(\xi+1)+
|z|\big)\,|s-t|^{\frac12}\nonumber\\
&\leq& D|z-\tilde{z}| + \tilde{F} (1+|z|)
|s-t|^{\frac12},\label{eq:pf-holder}
\end{eqnarray}
in which $\tilde{F}$ only depends on $T$ and $\xi$. Since $R$ is a
bounded domain, the factor $1+|z|$ in (\ref{eq:pf-holder}) is
bounded in $R$, so $z \rightarrow Pf(z,t)$ is Lipschitz and $t
\rightarrow Pf(z,t)$ is H\"{o}lder, uniformly with respect to the
other variable. Now by Theorem 5.2 in Chapter 6 of
\cite{friedman}, the parabolic partial differential equation
(\ref{eqn-Jf-1-pf}) and (\ref{eqn-Jf-2-pf}) has a unique classical solution
in the bounded domain $R$. Moreover this solution can be
represented by
\begin{eqnarray*}
u(z,t) &=& \E^{\Q} \left\{e^{-\lambda\xi \tau}Jf(Z_{\tau}^{t,z},
\tau+t)+ \int_0^{\tau}e^{-\lambda\xi s}\lambda\cdot Pf(Z_s^{t,z},
t+s)\, ds\right\} \\
&=& Jf(z,t),
\end{eqnarray*}
in which the exit time $\tau \triangleq
\inf_{s\in[0,T-t]}\{Z_{s}^{t,z} = z_1 \text{ or } z_2\}\wedge (T-t)$. The
second equality follows from the definition of the operator $J$ in
(\ref{eq:def-J-exp}) and the strong Markov property of $Z^{t,z}$.

So far we have shown that $Jf$ agrees with the unique classical solution of
(\ref{eqn-Jf-1-pf}) and (\ref{eqn-Jf-2-pf}) in the bounded domain
$R$. Since this statement holds for arbitrary $R$, it is clear
that $Jf$ is a solution of the parabolic partial differential
equation (\ref{eqn-Jf-1}) and (\ref{eqn-Jf-2}) for all $(z,t)\in
\R\times[0,T]$. The uniqueness of the solution follows from
Corollary 4.4 in Chapter 6 in \cite{friedman}, since the
coefficients of the derivative operators in (\ref{eq:def-At})
satisfy linear and quadratic growth conditions respectively.
\end{proof}

\subsection{Properties of the Sequence of Functions Defined in (\ref{eq:def-vn})}\label{sec:analyze-vn}

Our first goal is to prove $z \rightarrow v_n(z,t)$ is Lipschitz and $t \rightarrow v_n(z,t)$ is H\"{o}lder
continuous for all $n$.
To this end we will apply Lemmas~\ref {lemma:lip}
and \ref{lemma:holder}. To be able to apply the latter lemma we need to show that
\begin{equation}\label{eq:def-Mn}
M_n \triangleq  \sup_{t\in[0,T]} \{v_n(0,t)\}<\infty, \quad \text{for } n\geq 0.
\end{equation}
In the next lemma, we will dominate the sequence of
constants $(M_n)_{n\geq 0}$ by a universal constant $M_{\infty}$,
which depends only on $T$.

\begin{lemma}\label{lemma:Minf}
Let us define the sequence of constants $(M_n)_{n\geq 0}$ as in
(\ref{eq:def-Mn}), then
\[
 M_n < M_{\infty} \triangleq \frac{U}{1-\alpha} + \frac{\alpha}{1-\alpha}
 \frac{B}{\xi} + K_1< \infty ,
\]
in which the constants $U$, $B$ and $\alpha$ are defined in Lemma
\ref{lemma:mjf}.
\end{lemma}
\begin{proof}
When $n=0$, by the definition of $v_0(\cdot, \cdot)$ in
(\ref{eq:def-vn}), we have
\[
M_0 = \sup_{t\in[0,T]}v_0(0,t) = (\zeta\cdot(0-K_1))^+ \leq K_1,
\]
in which the last inequality is saturated when $\zeta=-1$. It
follows from Lemma \ref{lemma:mjf}
that
\begin{equation}
M_{n+1} \leq U + \alpha\left(M_n + \frac{B}{\xi}\right), \quad
\text{for } n \geq 0,
\end{equation}
in which $\alpha<1$. It can be proven by induction that
\begin{equation} \label{eq:ubMn}
M_n \leq U\left(\sum_{i=0}^{n}\alpha^i - \alpha^n\right) + \alpha
\left(\sum_{i=0}^{n}\alpha^i - \alpha^n \right)\frac{B}{\xi} + \alpha^n K_1,
\quad \text{for } n\geq 0.
\end{equation}
Since $U$, $B$ and $\xi$ are positive constants and $0<\alpha<1$,
it is clear from (\ref{eq:ubMn}) that
\[
 M_n \leq \frac{U}{1-\alpha} +
 \frac{\alpha}{1-\alpha}\frac{B}{\xi} + K_1 = M_{\infty} < \infty.
\]
\end{proof}

\begin{lemma}\label{lemma:vn}
Let $(v_n(\cdot,\cdot))_{n\geq 0}$ be as in (\ref{eq:def-vn}).
We have that
\begin{equation}\label{eq:vn-lip}
\left|v_n(z,t) - v_n(\tilde{z},t) \right|\leq
\left|z-\tilde{z}\right|, \quad z, \tilde{z} \in \R
\end{equation}
and
\begin{equation}
|v_n(z,t)-v_n(z,s)| \leq F_n (1+|z|)\,(s-t)^{\frac12}, \quad 0\leq
t < s \leq T, \label{eq:vn-holder}
\end{equation}
in which $F_n$ are all finite constants depending on $T$.
\end{lemma}
\begin{proof}
>From the definition of $v_0(\cdot,\cdot)$ in (\ref{eq:def-vn}), we
have
\begin{equation}
\left|v_0(z,t)-v_0(\tilde{z},t)\right| =
\left|(\zeta\cdot(z-K_1))^+ - (\zeta\cdot(\tilde{z}-K_1))^+\right|
\leq \left|z-\tilde{z}\right|. \label{eq:v0-lip}
\end{equation}
Now, the inequality (\ref{eq:vn-lip}) follows from induction and
Lemma \ref{lemma:lip}. On the other hand, (\ref{eq:vn-holder})
holds as a result of Lemma \ref{lemma:holder} which we can apply as a result of Lemma~\ref{lemma:Minf}.
\end{proof}

As a corollary of Remark \ref{remark-tildeDE} and Lemma
\ref{lemma:Minf}, we can show that
$(v_n(z,t))_{n\geq 0}$ satisfies a linear growth condition in the $z$-variable, uniformly in the $t$-variable. This will be used to show that this sequence has a limit.

\begin{corollary}\label{cor:dominatefun}
For any $n\geq 0$,
\begin{equation}\label{eq:def-L}
v_n(z,t) \leq M_{\infty} + |z| \triangleq L(z), \quad (z,t)\in
\R\times[0,T].
\end{equation}
\end{corollary}
\begin{proof}
An induction argument using the inequality (\ref{eq:v0-lip}), Lemma \ref{lemma:lip} and Remark \ref{remark-tildeDE}
gives
\[
v_n(z,t) \leq M_n + |z|, \quad \text{for } n\geq 0.
\]
Now, the result follows from Lemma \ref{lemma:Minf}.
\end{proof}

As a result of Corollary~\ref{cor:dominatefun}, next we show that, for a fixed $(z,t)\in \R\times [0,T]$, the sequence
$\{v_n(z,t)\}_{n\geq 0}$ is a Cauchy sequence.

\begin{lemma}\label{lemma:cauchy}
For any $(z,t)\in \R\times [0,T]$ and $n,m\geq 0$.
\begin{equation}
 \left|v_{n+m}(z,t)-v_m(z,t)\right| \leq 2 M_{\infty} A^m + 2\left(\frac{1}{rT} \frac{\xi+1}{\xi} +
 \frac{C}{\xi}\right)\left[\sum_{i=0}^{m} A^{m-i}B^i -B^m\right] +
 2|z|B^m,\label{eq:cauchy}
\end{equation}
where $A=1-e^{-\lambda\xi(T-t)}$, $B=1-e^{-\lambda(T-t)}$ and $C$
is the same constant as in (\ref{eq:eZ}).
\end{lemma}
\begin{proof}
We will prove the estimation (\ref{eq:cauchy}) by induction on
$m$. When $m=0$, it follows from Corollary \ref{cor:dominatefun}
that
\[
 |v_n(z,t)-v_0(z,t)|\leq 2M_{\infty} + 2|z|.
\]
It is clear that (\ref{eq:cauchy}) is satisfied in this case.
Assuming (\ref{eq:cauchy}) holds for $m$ case,  we will show that it holds when we replace $m$ by $m+1$. From the definition of $\{v_n(\cdot,\cdot)\}_{n\geq
0}$, we have
\begin{equation*}
 |v_{n+m+1}(z,t)-v_{m+1}(z,t)| \leq \E^{\Q} \left\{\int_0^{T-t} ds\, e^{-\lambda\xi s} \lambda\cdot
 \left|Pv_{n+m}(Z_s^{t,z},t+s)-Pv_m(Z_s^{t,z},t+s)\right|\right\}.
\end{equation*}
In the right hand side of above inequality, the induction assumption
gives us
\begin{eqnarray}
&&\left|Pv_{n+m}(Z_s^{t,z},t+s)-Pv_m(Z_s^{t,z},t+s)\right|\nonumber\\
&&\leq \int_{\R_+}
\left|v_{n+m}\left(\frac{Z_s^{t,z}}{y}+q_{t+s}\frac{y-1}{y},t+s\right)-v_m\left(\frac{Z_s^{t,z}}{y}+q_{t+s}\frac{y-1}{y},t+s\right)\right|y\nu(dy)\nonumber\\
&&\leq 2\xi M_{\infty}\left(1-e^{-\lambda\xi(T-t-s)}\right)^m\nonumber\\
&&\quad
+2\xi\left(\frac{1}{rT}\frac{\xi+1}{\xi}+\frac{C}{\xi}\right)
\left[ \sum_{i=0}^{m} \left(1-e^{-\lambda(T-t-s)}\right)^i
\left(1-e^{-\lambda\xi(T-t-s)}\right)^{m-i} - \left(1-e^{-\lambda(T-t-s)}\right)^m \right] \nonumber\\
&&\quad +2\int_{\R_+} \left(\frac{|Z_s^{t,z}|}{y}+q_{t+s}
\frac{|y-1|}{y}\right)\left(1-e^{-\lambda(T-t-s)}\right)^m
y\nu(dy)\nonumber\\
&&\leq 2\xi M_{\infty} (1-e^{-\lambda\xi(T-t)})^m \nonumber\\
&&\quad + 2\xi\left(\frac{1}{rT}\frac{\xi+1}{\xi} +
\frac{C}{\xi}\right) \left[ \sum_{i=0}^{m}\left(1-e^{-\lambda(T-t)}\right)^i
\left(1-e^{-\lambda\xi(T-t)}\right)^{m-i} - \left(1-e^{-\lambda(T-t)}\right)^m \right]\nonumber\\
&&\quad +2|Z_s^{t,z}|\left(1-e^{-\lambda(T-t)}\right)^m +
\frac{2}{rT}\left(\xi+1\right)\left(1-e^{-\lambda(T-t)}\right)^m.
\label{eq:est-ind}
\end{eqnarray}
In (\ref{eq:est-ind}), the third inequality follows, because $q_{t+s}\leq
\frac{1}{rT}$, and for $m\geq 1$
\[
 \sum_{i=0}^{m-1} \left(1-e^{-\lambda(T-t-s)}\right)^i \left(1-e^{-\lambda\xi(T-t-s)}\right)^{m-i} \leq  \sum_{i=0}^{m-1} \left(1-e^{-\lambda(T-t)}\right)^i \left(1-e^{-\lambda\xi(T-t)}\right)^{m-i},
\]
since $s\geq 0$.

On the other hand, from (\ref{sol z-t}), we have
\[|Z_s^{t,z}|\leq |z|H_{s}^0 + |b_{s}|,\]
where $\E^{\Q} \{|b_{s}|\}= \E^{\Q} \{|Z_s^{t,0}|\}\leq C$
from  (\ref{eq:eZ}). Therefore we have
\begin{equation}\label{eq:est-ind-2}
 \E^{\Q}\{|Z_s^{t,z}|\}\leq |z|e^{\mu s}+C.
\end{equation}
Taking expectation on both side of (\ref{eq:est-ind}) and plugging
(\ref{eq:est-ind-2}) back into (\ref{eq:est-ind}), we have
\begin{eqnarray}
&&\E^{\Q}\left|Pv_{n+m}(Z_s^{t,z},t+s)-Pv_m(Z_s^{t,z},t+s)\right|\nonumber\\
&&\leq 2\xi M_{\infty}  A^m
+2\xi\left(\frac{1}{rT}\frac{\xi+1}{\xi}+
\frac{C}{\xi}\right) \left[\sum_{i=0}^{m}A^{m-i}B^i-B^m\right]
\nonumber\\
&&\quad +2|z|e^{\mu s}B^m +
2\left(\frac{1}{rT}(\xi+1)+C\right)B^m. \label{eq:est-ind-3}
\end{eqnarray}
Multiplying both sides  of
(\ref{eq:est-ind-3}) with $e^{-\lambda\xi s}\lambda$ and integrating with respect to $s$ over
$[0,T-t]$, and using the identity $\mu-\lambda\xi=-\lambda$, we obtain the inequality
(\ref{eq:cauchy}) with $m$ replaced by $m+1$.
\end{proof}

As a result of the previous lemma we can define the pointwise
limit for the sequence $(v_n(\cdot,\cdot))_{n\geq 0}$:
\begin{equation}
 v_{\infty}(z,t) \triangleq \lim_{n\geq 0} v_n(z,t), \quad
 (z,t)\in \R\times[0,T].
\end{equation}

Moreover, as a corollary of Lemma \ref{lemma:cauchy}, we have
\begin{corollary}\label{cor:uniform-conv}
For any compact domain $\mathcal{D}\subset \R$, $v_n(z,t)$
converges uniformly to $v_{\infty}(z,t)$ for $(z,t)\in
\mathcal{D}\times[0,T]$. Moreover,
\begin{equation}\label{eq:uniform-conv}
|v_{\infty}(z,t)-v_n(z,t)|\leq
M_{\mathcal{D}}\left(1-e^{-\lambda\eta(T-t)}\right)^n,
\end{equation}
where $M_{\mathcal{D}}$ is a constant depending on $\mathcal{D}$
and $\eta=\max\{\xi,1\}$.
\end{corollary}
\begin{proof}
Observing that the right hand side of (\ref{eq:cauchy}) is
independent of $n$ and $|z|$ is uniformly bounded in
$\mathcal{D}$, the result follows from Lemma \ref{lemma:cauchy}.
\end{proof}

In the following, we will begin to study properties of
$v_{\infty}(\cdot,\cdot)$.
\begin{lemma}\label{lemma:fixpt}
The function $v_{\infty}$ is a fixed point of the
operator $J$.
\end{lemma}
\begin{proof}
For any $s \in [0,T-t]$,
\begin{eqnarray}
\E^{\Q}\left\{PL(Z_s^{t,z})\right\} &=& \E^{\Q}\left\{\int_{\R_+}
L\left(\frac{Z_s^{t,z}}{y} +
q_{t+s}\frac{y-1}{y}\right) y \nu(dy)\right\} \nonumber\\
&\leq& \E^{\Q}\left\{\int_{\R_+} \left[M_{\infty} +
\frac{|Z_s^{t,z}|}{y} +
q_{t+s}\frac{|y-1|}{y}\right] y \nu(dy)\right\} \nonumber\\
&\leq& \xi M_{\infty} + \frac{1}{rT}(\xi+1) + C(1+|z|). \label{eq:Lintegrable}
\end{eqnarray}
As a result, we have
\begin{eqnarray}
v_{\infty}(z,t) &=& \lim_{n\geq 0} v_{n+1}(z,t) \nonumber\\
&=& \lim_{n\geq
0}\E^{\Q}\left\{e^{-\lambda\xi(T-t)}\left(\zeta\cdot\left(Z_{T-t}^{t,z}-K_1\right)\right)^+
+ \int_0^{T-t} e^{-\lambda\xi s}\lambda \cdot (Pv_n)(Z_s^{t,z},t+s)\,
ds\right\}\nonumber\\
&=&\E^{\Q} \left\{e^{-\lambda\xi(T-t)}\left(\zeta\cdot\left(Z_{T-t}^{t,z}-K_1\right)\right)^+
+ \int_0^{T-t} e^{-\lambda\xi s}\lambda \cdot (P\lim_{n\geq
0}v_n)(Z_s^{t,z},t+s)\,
ds\right\}\nonumber\\
&=& Jv_{\infty}(z,t). \label{eq:Jfixpt}
\end{eqnarray}
 The third equality follows by applying
dominated convergence theorem three times. We can use the dominated convergence theorem due to Corollary \ref{cor:dominatefun} and
(\ref{eq:Lintegrable}).
\end{proof}

Using Lemmas \ref{lemma:vn} and (\ref{eq:uniform-conv}), we can
show $z \rightarrow v_{\infty}(z,t)$ is Lipschitz continuous and
$t \rightarrow v_{\infty}(z,t)$ is H\"{o}lder continuous.

\begin{lemma}\label{lemma:vinf-lip}
$v_{\infty}(\cdot,\cdot)$ satisfies
\begin{equation}\label{eq:vinf-lip}
\left|v_{\infty}(z,t)-v_{\infty}(\tilde{z},t)\right| \leq
|z-\tilde{z}|, \quad \text{for } (z,t), (\tilde{z},t) \in
\R\times[0,T].
\end{equation}
\end{lemma}
\begin{proof}
For fixed $z$ and $\tilde{z}$, let us choose a compact domain
$\mathcal{D}_{z,\tilde{z}} \subseteq \R$, so that $z,\tilde{z} \in
\mathcal{D}_{z,\tilde{z}}$. Then we have
\begin{eqnarray}
\left|v_{\infty}(z,t) - v_{\infty}(\tilde{z},t)\right| &\leq&
|v_{\infty}(z,t)-v_n(z,t)| + |v_n(z,t)-v_n(\tilde{z},t)| +
|v_n(\tilde{z},t)-v_{\infty}(\tilde{z},t)|\nonumber\\
&\leq& 2\left(1-e^{-\lambda\eta(T-t)}\right)^n
M_{\mathcal{D}_{z,\tilde{z}}} + |z-\tilde{z}|.
\label{est-vinf-lip}
\end{eqnarray}
In order to obtain the last inequality, we use Lemmas
\ref{lemma:vn} and Corollary \ref{cor:uniform-conv}. Since $n$ in
the second inequality in (\ref{est-vinf-lip}) is arbitrary, the
result follows.
\end{proof}

\begin{corollary} \label{lemma:vinf-holder}
$v_{\infty}(\cdot,\cdot)$ satisfies
\begin{equation}\label{eq:vinf-holder}
|v_{\infty}(z,t) - v_{\infty}(z,s)| \leq F_{\infty} \, (1+|z|)\,
|t-s|^{\frac12},
\end{equation}
in which constant $F_{\infty}< \infty$ .
\end{corollary}
\begin{proof}
This is a direct application of Lemmas \ref{lemma:holder}, \ref{lemma:Minf} and
\ref{lemma:vinf-lip}.  Note that Lemma~\ref{lemma:Minf} is needed to show that $\sup_{t\in[0,T]} \{v_{\infty}(0,t)\}<\infty$, which is required by Lemma~\ref{lemma:holder}.
\end{proof}

\subsection{Proof of Theorem~\ref{theorem-2}}\label{sec:proof}

\indent (i) This is a direct consequence of Lemma~\ref{lemma:cauchy}, which shows that the sequence
$\{v_n(z,t)\}_{n\geq 0}$ is a Cauchy sequence.

\noindent (ii) See Corollary~\ref{cor:uniform-conv}.

\noindent (iii) 
Using the inequalities (\ref{eq:vn-lip}) and
(\ref{eq:vn-holder}) in Lemma \ref{lemma:vn}, we can apply Proposition
\ref{theorem-1} to the function $f =v_n$. It indicates
$Jv_n(\cdot,\cdot)$ is the unique classical solution of the following
equation
\begin{eqnarray}
& &\A(t)Jv_n(z,t) -\lambda\xi Jv_{n}(z,t) + \lambda \cdot
(Pv_n)(z,t) +
 \frac{\partial}{\partial t}Jv_{n}(z,t) =0 \label{eqn-Jvn-1}\\
 & &Jv_n(z,T)= (\zeta\cdot(z-K_1))^+ ,\nonumber
\end{eqnarray}
for $(z,t)\in\R\times[0,T]$. By the definition of the sequence
$(v_n(\cdot,\cdot))_{n\geq0}$ in (\ref{eq:def-vn}), we have
$Jv_n(\cdot,\cdot)=v_{n+1}(\cdot,\cdot)$. So $v_{n+1}$ is the
unique solution of (\ref{eqn-vn-1}) and (\ref{eqn-vn-2}).

\noindent (iv) Because of Lemma
\ref{lemma:vinf-lip} and Corollary \ref{lemma:vinf-holder}, we can
apply Proposition \ref{theorem-1} to the function $f= v_{\infty}$. It
shows $Jv_{\infty}(\cdot,\cdot)$ is the unique classical solution of the
following parabolic partial differential equation
\begin{eqnarray}
& &\A(t)Jv_{\infty}(z,t) -\lambda\xi Jv_{\infty}(z,t) + \lambda
\cdot (Pv_{\infty})(z,t) +
 \frac{\partial}{\partial t}Jv_{\infty}(z,t) =0 \label{eqn-Jvinf-1}\\
 & &Jv_{\infty}(z,T)= (\zeta\cdot(z-K_1))^+ ,\label{eqn-Jvinf-2}
\end{eqnarray}
However, $Jv_{\infty}=v_{\infty}$ by Lemma
\ref{lemma:fixpt}. Therefore, $v_{\infty}(\cdot,\cdot)$ is the unique
classical solution of the integro-partial differential equation
(\ref{eqn-vinf-1}) and (\ref{eqn-vinf-2}).
\hfill $\square$
\newline
\newline
\textbf{Acknowledgment} We are grateful to the anonymous associate editor and the two referees for detailed comments that helped us improve our paper.


\begin{table}[p!b!]
\begin{minipage}[c]{\textwidth}
\label{table:kou} \caption{The approximated price for continuously
averaged European type Asian options for a double exponential jump model.}
$r=0.15$, $S_0=100$, $T=1$, $p=0.6$ and $\eta_1=\eta_2=25$. Monte Carlo method
uses $10^6$ simulations and $10^3$ time steps.``C - P" is the
difference between our approximated call and put option prices. ``Parity" is the difference predicted by the put-call parity. Run
times are in seconds.

\begin{center}
\begin{tabular}{|r|r|r|r|r|r|r|r|r|r|r|r|}
\hline

\multicolumn{12}{|c|}{European Asian call option prices for a double exponential jump diffusion model}\\

\hline

\multirow{3}{*}{$\sigma$} & \multirow{3}{*}{$K_2$} &
\multirow{3}{*}{$\lambda$} & \multicolumn{6}{|c|}{Iteration
Algorithm} &
\multicolumn{3}{|c|}{Monte Carlo (Call Option)}\\
\cline{4-9}

& & & \multicolumn{2}{|c|}{Call Option (C)} &
\multicolumn{2}{|c|}{Put Option (P)} & \multirow{2}{*}{C - P} & \multirow{2}{*}{Parity} & \multicolumn{3}{|c|}{} \\

\cline{4-7} \cline{10-12}

& & & Value & Time  & Value & Time  & & & Value & Std. Err. & Time \\

\hline

\multirow{6}{*}{0.1} & \multirow{2}{*}{90} & 1 & 15.419 & 1.0 & 0.012 & 1.0 & 15.407 & \multirow{2}{*}{15.398} & 15.410 & 0.006 & 913  \\
 &  & 3 & 15.457 & 1.5 & 0.045 & 1.5 & 15.412 & & 15.442 & 0.007 & 976 \\
\cline{2-12}

 & \multirow{2}{*}{100} & 1 & 7.170 & 1.0 & 0.376 & 1.0 & 6.794 & \multirow{2}{*}{6.791} & 7.170 & 0.006 & 919 \\
 &  & 3 & 7.456 & 1.5 & 0.656 & 1.6 & 6.800 & & 7.439 & 0.007 & 987 \\
\cline{2-12}

 & \multirow{2}{*}{110} & 1 & 1.702 & 1.0 & 3.520 & 1.0 & -1.818 & \multirow{2}{*}{-1.817} & 1.697 & 0.004 & 906 \\
 &  & 3 & 2.220 & 1.5 & 4.040 & 1.6 & -1.820 & & 2.207 & 0.004 & 981 \\
\hline

\multirow{6}{*}{0.2} & \multirow{2}{*}{90} & 1 & 15.699 & 1.0 & 0.292 & 1.0 & 15.407 & \multirow{2}{*}{15.398} & 15.686 & 0.012 & 908 \\
 & & 3 & 15.802 & 1.5 & 0.390 & 1.6 & 15.412 &  & 15.806 & 0.012 & 983 \\

\cline{2-12}

 & \multirow{2}{*}{100} & 1 & 8.540 & 1.0 & 1.745 & 1.0 & 6.795 & \multirow{2}{*}{6.791} & 8.540 & 0.010 & 935 \\
 & & 3 & 8.790 & 1.5 & 1.994 & 1.6 & 6.796 & & 8.784 & 0.010 & 996\\

\cline{2-12}

 & \multirow{2}{*}{110} & 1 & 3.723 & 1.0 & 5.541 & 1.0 & -1.818 & \multirow{2}{*}{-1.817}& 3.721 & 0.007 & 921\\
 & & 3 & 4.045 & 1.6 & 5.864 & 1.6 & -1.819 & & 4.038 & 0.007 & 983 \\
\hline
\end{tabular}
\end{center}

\end{minipage}



\begin{minipage}[c]{\textwidth}
\label{table:merton} \caption{The approximated price for
continuously averaged European type Asian options for normal jump diffusion model.}
$r=0.15$, $S_0=100$, $T=1$, $\lambda=1$, $\tilde{\mu}=-0.1$ and
$\tilde{\sigma}=0.3$. Monte Carlo method uses $10^6$ simulations
and $10^3$ time steps. ``C - P" is the difference between our
approximated call and put option prices. ``Parity" is the
difference predicted by the put-call parity. Run times are in
seconds.

\begin{center}
\begin{tabular}{|r|r|r|r|r|r|r|r|r|r|r|}
\hline

\multicolumn{11}{|c|}{European Asian call option with normal jump}\\

\hline

\multirow{3}{*}{$\sigma$} & \multirow{3}{*}{$K_2$} & \multicolumn{6}{|c|}{Iteration Algorithm} & \multicolumn{3}{|c|}{Monte Carlo (Call Option)}\\
\cline{3-8}

& & \multicolumn{2}{|c|}{Call Option} & \multicolumn{2}{|c|}{Put
Option} & \multirow{2}{*}{C - P} & \multirow{2}{*}{Parity} &
\multicolumn{3}{|c|}{} \\

\cline{3-6} \cline{9-11}

& & Value & Time & Value & Time & & & Value & Std. Err. & Time\\

\hline

\multirow{3}{*}{0.1} & 90 & 16.997 & 0.5 & 1.601 & 0.5 & 15.396 & 15.398 & 16.991 & 0.014 & 913 \\
 & 100 & 10.062 & 0.5 & 3.272 & 0.5 & 6.789 & 6.791 & 10.046 & 0.013 & 910 \\
 & 110 & 4.836 & 0.5 & 6.653 & 0.5 & -1.817 & -1.816 & 4.834 & 0.011 & 915 \\

\hline

\multirow{3}{*}{0.2} & 90 & 17.346 & 0.5 & 1.950 & 0.5 & 15.396 & 15.398 & 17.339 & 0.017 & 919\\
 & 100 & 10.959 & 0.5 & 4.170 & 0.5 & 6.789 & 6.791 & 10.968 & 0.015 & 917\\
 & 110 & 6.303 & 0.5 & 8.120 & 0.5 &-1.817 & -1.816 & 6.310 & 0.012 & 913 \\

\hline
\end{tabular}
\end{center}
\end{minipage}

\end{table}

\begin{table}[t!b!p!]
The parameters for Asian options in the following three tables are
the same as the parameters used in the 7th row in Table 1, i.e. $r=0.15$, $S_0=100$,
$K_1=0$, $K_2=90$, $T=1$, $\sigma=0.2$, $\lambda=1$, $p=0.6$ and
$\eta_1=\eta_2=25$.

\begin{minipage}[c]{\textwidth}
\label{table:num-tru} \caption{The convergence of the option price
with respect to the truncation length of the numerical integral.}
As we introduced in (\ref{eq:numint}), the integral term in
(\ref{eq:op-P}) is approximated by the trapezoidal rule on an
interval $[x_{min}, x_{max}]$ with $x_{min} = x_0 < x_1 < \cdots <
x_L = x_{max}$. In this table, fixing the discretization, we study
the convergence with respect to the length of the truncation
interval $[x_{min}, x_{max}]$. We choose $x_{min} = -N/\eta_2$ and
$x_{max} = N/\eta_1$. In (\ref{eq:op-P-log}), if the distribution
$F$ is the double exponential, when $N$ is large, the probability
that the random variable $X$ be outside the interval $[-N/\eta_2,
N/\eta_1]$ is very small (for example, when $N=15$ the probability
is less than $10^{-6}$).

\begin{center}
\begin{tabular}{|r|r|r|r|r|r|}
\hline

\multicolumn{6}{|c|}{Convergence with respect to truncation}\\

\hline

N & Call Option (C)& Time & Put Option (P) & Time & (C - P) - Parity\\

\hline

5 & 15.5832 & 0.500 & 0.2858 & 0.516 &  -0.1002\\
8 & 15.6953 & 0.765 & 0.2916 & 0.797 & 0.0061\\
10 & 15.6994 & 0.969 & 0.2921 & 1.000 &  0.0097\\
12 & 15.6995 & 1.141 & 0.2921 & 1.187 &  0.0098\\
15 & 15.6995 & 1.391 & 0.2921 & 1.500 &  0.0098\footnote{Because we fix the discretization of the numerical integral, the difference between the calculated value and predicted value in the last column doesn't seem to converge to 0. But as $\Delta x \rightarrow 0$, the difference will converge to 0 as we will see in the next Table.}\\

\hline
\end{tabular}
\end{center}
\end{minipage}

\hspace{2cm}

\begin{minipage}[c]{\textwidth}
\label{table:num-dis} \caption{The convergence of the option price
with respect to the grid size of the numerical integral.} In this
table, we fix the truncation of the numerical integral as
$x_{min}= -10/\eta_2$ and $x_{max}=10/\eta_1$, we will show the
convergence with respect to the number of grids $L$ in the
discretization of numerical integral in (\ref{eq:numint}). Since
the density of double exponential distribution has a cusp at zero,
we choose an unequaly spaced grid here. The closer $x$ to zero is, the
finer the grid is. While fixing the truncation interval $[x_{min},
x_{max}]$, the larger $L$ is the finer the grid is (see Table 4
for the notation).

\begin{center}
\begin{tabular}{|r|r|r|r|r|r|}
\hline

\multicolumn{6}{|c|}{Convergence with respect to discretization}\\

\hline

L & Call Option (C) & Time & Put Option (P) & Time &  (C - P) - Parity \\

\hline

200 & 15.7295 & 0.422 & 0.2926 & 0.422  & 0.0393\\
300 & 15.7103 & 0.578 & 0.2923 & 0.609  & 0.0204\\
400 & 15.7034 & 0.766 & 0.2924 & 0.797  & 0.0134\\
500 & 15.6944 & 0.969 & 0.2921 & 1.000  & 0.0097\\
600 & 15.6968 & 1.141 & 0.2920 & 1.172  & 0.0072\\
700 & 15.6954 & 1.344 & 0.2920 & 1.360  & 0.0058\\
800 & 15.6943 & 1.516 & 0.2920 & 1.562  & 0.0047\\

\hline
\end{tabular}
\end{center}
\end{minipage}

\end{table}

\begin{table}[t!b!p!]
\begin{minipage}[c]{\textwidth}
\label{table:num-grid} \caption{The convergence of the option
price with respect to the grid sizes used in the finite difference
scheme.} In this table we fix $x_{min}= -10/\eta_2$ and
$x_{max}=10/\eta_1$,  $L=1000$ (See Table 4 for the notation).
Moreover, we fix $z_{min} = z-0.5 $ and $z_{max}= z+0.5$ with $z=
\left(1-e^{-rT}\right)/(rT)-e^{-rT}K_2/S_0$ defined in
(\ref{eq:defn-z-x0-s0}). We will show the convergence with respect
to time and space grid sizes that are used in implementing the
finite difference scheme.

\begin{center}
\begin{tabular}{|r|r|r|r|r|}
\hline

\multicolumn{5}{|c|}{Convergence with respect to grid sizes}\\

\hline

Number of Time Steps & Number of Space Steps & Call Option Price & Changes  & Time\\

\hline

10 & 40 & 15.7093 & n.a. & 0.438\\
25 & 100 & 15.6929 & 0.0164 & 1.890\\
50 & 200 & 15.688 & 0.0049 & 7.500\\
100 & 400 & 15.6864 & 0.0016 & 29.406\\

\hline
\end{tabular}
\end{center}
\end{minipage}
\end{table}

\end{document}